\documentclass[11pt]{article}

\usepackage{amsthm}

\newtheorem{theorem}{Theorem}
\newtheorem{lemma}[theorem]{Lemma}
\newtheorem{corollary}[theorem]{Corollary}

\theoremstyle{remark}
\newtheorem{remark}[theorem]{Remark}

\usepackage{authblk}
\usepackage{graphicx,physics,amsmath,amssymb,graphics,setspace,float} 
\usepackage[colorlinks=true, linkcolor=blue, citecolor=blue, linktocpage=true]{hyperref}
\usepackage[margin=1in,letterpaper]{geometry}
\usepackage[numbers, sort&compress]{natbib}
\newcommand{\ZZ}[0]{\mathbb{Z}}
\usepackage{simpler-wick}
\newcommand{\szego}[0]{Szeg\H{o} }
\newcommand{\p}[0]{P_{\Omega}}

\newcommand{\q}[0]{P_\Omega^\perp}
\title{An elementary approach to the \szego theorems}
\author[1]{\small Israel Klich}
\author[1]{\small Ravisankar Rajagopal}
\affil{\textit{\small Department of Physics, University of Virginia,
Charlottesville, Virginia 22904-4714, USA}}
\begin{document}
\date{}

\maketitle
\begin{abstract}
\noindent
The strong \szego theorem and the Fisher-Hartwig formula describe the
asymptotic behavior of determinants obtained by restricting
translationally invariant operators to large finite regions. In particular, Szeg\H{o}'s theorem states that, beyond the leading extensive contribution to such a determinant, there is a universal finite correction,
$\sum_{r>0} rV_rV_{-r}$
where $V_r$ are the Fourier coefficients of the logarithm of the symbol. Although this formula appears throughout analysis and mathematical physics, the origin of its quadratic form and of the weight $r$ is often obscured in its standard analytic derivations. We develop an elementary picture in which the strong \szego term arises from \emph{boundary leakage}: paths contributing to the matrix exponential leave the finite interval and subsequently return. Counting the shortest such excursions immediately produces the factor $r$. We then show that an isolated boundary contributes exactly one half of the strong \szego correction, while the two boundaries of a long interval decouple. For finite-range matrices this gives a rigorous and intuitive elementary derivation of the asymptotics, with extensions to sufficiently rapidly decaying interactions. A complementary free-fermion argument provides a physical interpretation of the same cancellation mechanism. Numerical examples illustrate the asymtotics in the well controlled regime, the Fisher Hartwig logarithmic correction. The boundary formulation does not itself require translation invariance and therefore suggests natural extensions to disordered systems, higher dimensions, and symbols with internal matrix structure. To illustrate this we numerically go beyond the formal applicability region of the \szego formula, considering slow decay and disordered cases, and showing in some cases a robustness of the approach even for slowly decaying but disordered potentials. 
\end{abstract}

\section{Introduction}
\label{sec:introduction}

A Toeplitz matrix is a matrix whose entries depend only on the
difference between their indices,
\begin{equation}
    M_{jk}=c_{j-k}.
\end{equation}
The coefficients $c_r$ may be regarded as the Fourier coefficients of
a function $f$ on the unit circle,
\begin{equation}
    c_r
    =
    \frac{1}{2\pi}
    \int_0^{2\pi}
    f(e^{i\theta})e^{-ir\theta}\,d\theta.
\end{equation}
$f$ is referred to as the ``symbol" and the $N\times N$ Toeplitz matrix associated with $f$ is denoted by
$T_N(f)$. We are interested in its determinant
\begin{equation}
    D_N(f):=\det T_N(f).
    \label{eq:toeplitz-determinant-definition}
\end{equation}
To connect the symbol with the operator notation used below, write
\begin{equation}
V(e^{i\theta})=\log f(e^{i\theta})
=
\sum_{r\in\mathbb Z}V_r e^{ir\theta},
\label{eq:log-symbol}
\end{equation}
where a branch of the logarithm has been fixed.  There is a corresponding translationally invariant operator 
(Laurent operator on $\ell^2(\mathbb Z)$):
\begin{equation}
L(V)=\sum_{r\in\mathbb Z}V_rS^r,
\qquad
[L(V)]_{nm}=V_{n-m},
\label{eq:Laurent-operator-definition}
\end{equation}
where $S|j\rangle=|j+1\rangle$.  Under Fourier transformation,
$L(V)$ becomes multiplication by the function $V(e^{i\theta})$.
Consequently, the operator functional calculus gives
\begin{equation}
e^{L(V)}=L(e^V)=L(f).
\label{eq:Laurent-exponential}
\end{equation}
We henceforth suppress the symbol $L$ and write simply $V$ for
$L(V)$.  Thus
\begin{equation}
T_N(f)=P_Ne^VP_N,
\qquad
D_N(f)=\det\!\left(P_Ne^VP_N\right).
\label{eq:Toeplitz-as-compressed-exponential}
\end{equation}
The central point is that cutting to the finite interval and taking
the exponential do not commute:
\begin{equation}
P_Ne^VP_N\neq e^{P_NVP_N}
\label{eq:cutting-exponentiating-noncommute}
\end{equation}
in general.  The difference between these two operations is precisely
the boundary leakage studied below.

Toeplitz determinants arise naturally whenever a translationally
invariant problem is observed only within a finite region. In
statistical mechanics they appear in partition functions and
correlation functions of solvable lattice models, while in quantum
physics they occur in free-fermion overlaps, full counting statistics,
and calculations of entanglement entropy. They also play a central
role in orthogonal polynomials, random matrix theory, and spectral
analysis
\cite{szego1952,boettcher1990,deift1999}. 

The Szeg\H{o} limit theorems describe the large-$N$ behavior of
$D_N(f)$. The first, or weak, Szeg\H{o} theorem determines the
extensive contribution:
Let $f$ be a positive function on the unit circle such that
$\log f\in L^1(\mathbb T)$. Then the weak \szego theorem states,
\begin{equation}
    \lim_{N\to\infty}
    \frac{1}{N}\log D_N(f)
    =
    \frac{1}{2\pi}
    \int_0^{2\pi}
    \log f(e^{i\theta})\,d\theta.
    \label{eq:weak-szego}
\end{equation}
This leading result goes back to Szeg\H{o}'s original work
\cite{szego1915}.

To state the strong theorem, write
\begin{equation}
    v(\theta)=\log f(e^{i\theta}),
    \qquad
    V_r=
    \frac{1}{2\pi}
    \int_0^{2\pi}
    v(\theta)e^{-ir\theta}\,d\theta.
    \label{eq:fourier-log-symbol}
\end{equation}
Under the Szeg\H{o}-Krein condition
\begin{equation}
    \sum_{r\in\mathbb Z}|r|\,|V_r|^2<\infty,
    \label{eq:szego-krein-condition}
\end{equation}
the strong Szeg\H{o} theorem states that
\begin{equation}
    \log D_N(f) =
    NV_0+\sum_{r=1}^{\infty}rV_rV_{-r}+o(1).
    \label{thm:szego1}
\end{equation}
For a positive real symbol, $V_{-r}=\overline{V_r}$, but it is useful
to retain the more general product $V_rV_{-r}$. Rigorous versions and
extensions of this result were developed over many decades
\cite{szego1952,widom1973,boettcher1990,simon2011}. 

When the symbol is not smooth, the strong \szego formula generally has to be replaced by the Fisher-Hartwig asymptotics \cite{FisherHartwig1969,HartwigFisher1969}.  The Fisher-Hartwig formula applies to symbols containing root singularities or phase jumps.  A singularity with parameters $\alpha_j$ and $\beta_j$ produces a logarithmic contribution
$(\alpha_j^2-\beta_j^2)\log N$ to $\log D_N(f)$
\cite{boettcherSilbermann1985,deift2011asymptotics}.
Such singularities arise naturally in gapless free-fermion systems, where the sharp Fermi surface produces jump discontinuities in the occupation-number symbol.  The Fisher-Hartwig formula has found particularly important applications in quantum-information and many-body physics.  It was first used to obtain the asymptotic entanglement entropy of blocks in free-fermion and quantum spin chains \cite{jin2004,its2005entanglement,keatingMezzadri2004}, and was subsequently applied to R\'enyi and single-copy entanglement, finite-size corrections, and boundary effects \cite{eisertCramer2005,eisertOsborne2006,franchiniItsKorepin2008,itsKorepin2009,calabreseEssler2010,fagottiCalabrese2011}.  Closely related Toeplitz determinants also govern fermionic full counting statistics \cite{abanovIvanovQian2011}.  More recently, extensions of these ideas to discontinuous block-Toeplitz symbols have been used to study entanglement in long-range fermionic chains and symmetry-resolved entanglement \cite{aresEtAl2018,aresMurcianoCalabrese2022}.  Long-range systems can also produce symbols lying outside the domains of both the strong \szego theorem and the conventional Fisher-Hartwig formula, leading, for example, to sublogarithmic entanglement growth \cite{aresEtAl2019}.  In higher dimensions, the corresponding role is played by Widom-type asymptotics, which determine the logarithmic violation of the fermionic entanglement area law \cite{gioevKlich2006,calabrese2012,leschkeSobolevSpitzer2014}.

The leading term in Eq.~\eqref{thm:szego1} has a simple interpretation:
Far from the ends of the interval, the system is effectively
translation invariant and can be diagonalized by plane waves where $\theta$ is the momentum. The
logarithm of the determinant is then the sum of the logarithms of the
corresponding eigenvalues, giving the bulk average $NV_0$.

The strong Szeg\H{o} term is considerably less transparent. Standard
proofs establish it through operator factorizations and asymptotic
identities, but these methods do not immediately reveal the mechanism
behind the formula. Why is the correction quadratic in the Fourier
coefficients of $\log f$? Why does the coefficient $V_rV_{-r}$ carry
the particular weight $r$? Why are there no cubic or higher-order
terms? Despite the maturity of the rigorous theory, a simple real-space
and physical explanation of these features is not usually emphasized.

It is also important that the relevant notion of range is the range
of $\log f$, rather than that of $f$ itself. A Toeplitz matrix
constructed from $f$ may be banded, while the Fourier coefficients
$V_r$ of $\log f$ decay only algebraically. In particular, when the
symbol has a zero or another singularity, one can have
$V_r\sim 1/r$. The sum in Eq.~\eqref{thm:szego1} then grows as
$\log N$, leading into the Fisher-Hartwig regime. Thus rapid decay of
the entries of $T_N(f)$ is not, by itself, an indication that the
subleading correction is finite.

The purpose of this paper is to give a direct boundary interpretation
of the strong Szeg\H{o} term. Let $V$ be the Laurent operator whose
symbol is $v=\log f$. The infinite Laurent operator with symbol $f$ is
then $e^V$, and the finite Toeplitz matrix can be written as the
restriction
\begin{equation}
    T_N(f)=P_Ne^VP_N,
\end{equation}
where $P_N$ projects onto an interval of $N$ lattice sites. The
difference between restricting $e^V$ and exponentiating the restricted
operator $P_NVP_N$ is entirely caused by paths that leave the interval
and return. We refer to these processes as boundary leakage.

In this paper, we consider $P_Ne^{t V}P_N$ and expanding the logarithm of the determinant in $t$ we find  that the first boundary
contribution is
\begin{equation}
    \frac{t^2}{2}
    \operatorname{Tr}
    \bigl(P_NV(1-P_N)VP_N\bigr).
\end{equation}
This expression describes a path that starts inside the interval,
crosses one of its boundaries, and then returns. A jump of length $r$
can cross a given endpoint from exactly $r$ starting sites. The
apparently mysterious factor $r$ in the strong Szeg\H{o} formula is
therefore simply the number of ways in which a jump of range $r$ can
cross the boundary.

Promoting this second-order observation to the full result requires
two further ingredients. First, the two ends of a long interval
decouple. Any term that involves both boundaries must propagate across
the entire interval and vanishes as the separation between the
boundaries tends to infinity. This decoupling is a consequence of
locality and does not require translation invariance. In the appendix
we give an elementary trace-norm proof for arbitrary bounded
finite-range matrices and explain how the argument extends to
sufficiently rapidly decaying matrices.

Second, for a scalar translation-invariant operator, the contribution
of one isolated boundary can be evaluated exactly. Splitting $V$ into
positive and negative shifts and using an elementary commutator
identity, we show that the logarithm of the single-boundary determinant
is exactly quadratic in the parameter $t$. Each endpoint contributes
one half of the strong Szeg\H{o} term. Combining this exact
single-interface result with the decoupling of the two endpoints gives
Eq.~\eqref{thm:szego1}. For finite-range $V$, this argument can be made
rigorous using elementary matrix estimates. Extending it to the full
Szeg\H{o}-Krein class requires additional control, for which
Hankel-operator methods or the Borodin-Okounkov identity provide the
standard tools.

We also present the boundary picture in the language of free
fermions. A related derivation of the Szeg\H{o} theorem using
the boson-fermion correspondence was given in
\cite{Hirayama:1990ma}. The underlying representation of determinants
as fermionic matrix elements, together with the current algebra and
boson-fermion correspondence, is part of the standard free-fermion
formalism 
\cite{DateJimboKashiwaraMiwa1983,MiwaJimboDate2000,AlexandrovZabrodin2013}. A finite
occupied interval has two Fermi edges. In the large-$N$ limit these
edges decouple, and each isolated edge produces one half of the strong
Szeg\H{o} correction. This gives a physical counterpart of the
elementary matrix argument and makes clear which part of the reasoning
continues to hold when translation invariance is broken.

Separating boundary decoupling from the exact evaluation of one
boundary also clarifies what remains true without translation
invariance. For a disordered finite-range matrix, the two boundaries
still decouple as their separation grows. The limiting correction is
therefore a sum of two local edge contributions. However, the special
scalar Toeplitz cancellation at each edge is generally absent, and
higher-order terms need not vanish. The strong Szeg\H{o} formula is
thus a particularly simple member of a broader class of boundary
factorization phenomena.

We illustrate these ideas numerically in three regimes. We first
consider algebraically decaying Fourier coefficients for which the
strong Szeg\H{o} constant is finite but convergence is slow. We then
study the borderline decay $V_r\sim1/r$, where the same boundary
counting predicts logarithmic growth. Finally, we examine both fixed
and ensemble-averaged disorder, demonstrating that boundary
decoupling survives even though the exact quadratic single-edge
formula does not. The goal is to isolate the elementary mechanism behind the formula
and to show which parts of that mechanism persist in more general
settings.

\section{\szego from an elementary Taylor expansion.}
Let $\Omega \subset \ZZ$ be an interval in $\ZZ$, and $\p$ be a projection to $\Omega$. For concreteness, we will take $\Omega = [0, N-1]$. 

Let the Toeplitz matrix be constructed out of the symbol $f(\theta) = e^{v(\theta)}$, and denote the matrix by $T_N(f)=P_N e^V P_N$, where $V_{m n} = V_{m-n}$ itself is a Toeplitz matrix constructed out of the symbol $v(\theta)$. We will be interested in the asymptotics of the determinant $T$, constrained to $\Omega$, when $N$ becomes large, i.e. in
\begin{equation}
 \det_{i,j\in\Omega} [e^{ V}]_{ij}
\end{equation}
It is convenient to add a parameter $t$ and denote the desired determinant by
\begin{equation}
   D_\Omega(t)
   :=\det_{i,j\in\Omega}[e^{tV}]_{ij}
   =\det(\q+\p e^{tV}\p),
   \qquad
   \log D_\Omega(t)=\tr\log(\q+\p e^{tV}\p).
\end{equation}
Our goal is to describe $\log D_\Omega(1)$. As a first step, we will simply expand in $t$, getting:
\begin{align}
    \log D_\Omega(t) &= t \tr (\p V \p) + \frac{t^2}{2} \tr (\p V^2 \p - (\p V \p)^2) + \order{t^3} \nonumber \\
    &= t \tr (\p V \p) + \frac{t^2}{2} \tr (\p V \q V \p) + \order{t^3} \label{eq:t_expansion}
\end{align}
The first order term is immediately the leading \szego term:
\begin{align}
    \tr(\p V \p) = \abs{\Omega} V_0.
\end{align}

Writing explicitly we can understand the second term as a sum over paths going from inside $\Omega$, jumping to a point outside $\Omega$ and then returning back to the initial point. In one dimension, we have, explicitly: \begin{figure}
    \centering
    \includegraphics[width=0.5\linewidth]{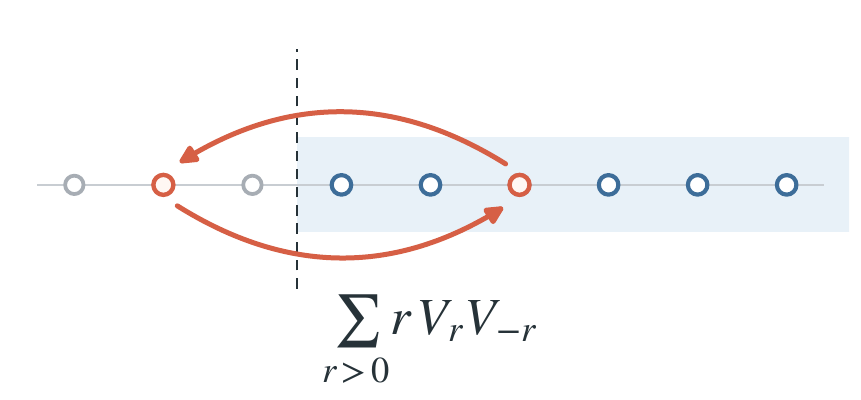}
    \caption{Origin of the strong \szego correction term: For a jump of length \(r\leq N\), exactly \(r\) starting sites near each endpoint allow the path to leave the interval and return.}
    \label{fig:szego fast decay clean}
\end{figure}
\begin{align}\label{eq:second-term-derivation}
    \tr(\p V \q V\p) = \sum_{n \in \Omega} \sum_{m \in \ZZ \setminus \Omega} V_{n-m} V_{m-n} 
\end{align}
any such path of length less than $1<r\leq N$ can have $r$ starting points inside $\Omega$, and paths of length $r>N$ have exactly $N$ possible starting points, leading to 
\begin{equation}
    \tr(\p V \q V\p) = 2 \qty(\sum_{r = 1}^{N} r V_r V_{-r} + N \sum_{r=N+1}^\infty V_r V_{-r}) \label{eq:FiniteNCorrection}
\end{equation}
where the factor $2$ comes from paths going to the left or to the right of the interval $\Omega$.

The second term in the brackets can be neglected in the limit of large $N$ if $V_r$ decay at least as fast as $1/{|r|^{1+ \epsilon}}$, and we get
\begin{equation}
    \log D_\Omega(t) = t \abs{\Omega} V_0 + t^2 \qty(\sum_{r=1}^\infty r V_r V_{-r} + o(1)) + \order{t^3}\label{eq:Szego exp}
\end{equation}
If we now set $t=1$, we have the strong \szego result. The mysterious factor $r$ is therefore just a count: it is the number of starting points from
which a jump of length $r$ crosses an endpoint. 

\subsection{Decoupling of the boundaries}
The argument above, leading to Eq. \eqref{eq:Szego exp}, is a demonstration that the strong \szego result originates from simply counting  matrix elements. It is however a Taylor expansion for small $t$, while the \szego formula has $t=1$.
A rigorous proof of the formula requires control of the next terms. 

The justification for retaining just the linear and quadratic in $t$ term hinges on two mechanisms: (1) when $N$ is large, the left and the right boundaries essentially decouple, and (2) the single boundary contribution, in the large $N$ limit, can be approximated by a semi-infinite interval for which all contributions beyond quadratic in $t$ vanish exactly.

The decoupling of the left and right boundaries is intuitive: one expects that for decaying $V$ any expansion involving the left and right boundaries of $\Omega$ would requires high powers of $V$ and would be less important. For the decoupling, the Toeplitz property, the translational invariance of $V$ is not needed.

In the appendix we outline a rigorous but elementary proof that decoupling the boundaries is justified for matrices $V$ that have finite range $R$ i.e. $V_{ij}=0$ whenever $\left|i-j\right|>R$. We also show that the determinant is approximated by contributions from two semi-infinite determinants representing the two boundaries up to factorially vanishing corrections in $N$, uniformly in $t\in [0,1]$.  The more general situation, where $V$ decays, but is not strictly compact can also be achieved. In the mathematical literature a convenient approach is based on the Borodin-Okounkov formula \cite{BorodinOkounkov2000,BasorWidom2000}. As we see below, the treatment of a single boundary can be done essentially exactly.
\section{One interface is exactly quadratic}
\label{sec:one-interface}

We first isolate a single boundary.  Let $S$ be shift operator 
$S\lvert j\rangle=\lvert j+1\rangle$, and let
$P_-$ project onto the negative half-line
$\mathbb Z_- = \{-1,-2,\ldots\}$.  For a translation-invariant scalar operator (a scalar Laurent operator), write
\begin{equation}
 V=V_0 I+V_++V_-,
 \qquad
 V_+=\sum_{r>0}V_rS^r,
 \qquad
 V_-=\sum_{r>0}V_{-r}S^{-r}.
 \label{eq:laurent-splitting}
\end{equation}
Set $A=P_-V_+P_-$ and $B=P_-V_-P_-$, regarded as operators on
$P_-\ell^2(\mathbb Z)$.  Since the shifts in $V$
commute before the boundary is introduced, their exponentials factorize,
and triangularity gives
\begin{equation}
 P_-e^{tV}P_-=P_-e^{tV_0}e^{tP_-V_+P_-}e^{tP_-V_-P_-}P_-
 \label{eq:half-line-factorization}
\end{equation}
To see this, note that $V_-$ translates to the left so starting with points that belong to $P_-$ will always stay in $P_-$, therefore  $e^{V_-}P_-=e^{P_-V_-P_-}P_-$. Similarly, $P_- e^{V_+}=P_- e^{P_-V_+P_-}$. 
The contribution to the determinant from the boundary is therefore,
\begin{equation}
 D_-(t)
 =\det\!\left[(P_-e^{tV}P_-)e^{-tP_-VP_-}\right]
 =\det\!\left(e^{tA}e^{tB}e^{-t(A+B)}\right).
 \label{eq:relative-edge-determinant}
\end{equation}
For finite-range $V$, the commutator $[A,B]$ is finite rank, so the relative
determinant in \eqref{eq:relative-edge-determinant} is well defined (it is an ordinary Fredholm
determinant).  A short matrix calculation (which is a special case
of a Helton-Howe-Pincus identity\cite{Pincus1968,HeltonHowe1973}), gives
\begin{align}
 \frac{d}{dt}\log D_-(t)
 &=\operatorname{Tr}\!\left(e^{-tB}Ae^{tB}-A\right) \notag\\
 &=\int_0^t
   \operatorname{Tr}\!\left(e^{-sB}[A,B]e^{sB}\right)\,ds
 =t\operatorname{Tr}[A,B].
 \label{eq:edge-derivative}
\end{align}
In the last step we used cyclicity only with the trace-class operator
$[A,B]$.  Thus no illicit use of $\operatorname{Tr}(AB)=\operatorname{Tr}(BA)$
is involved: $AB$ and $BA$ need not separately be trace class.

The remaining trace can be read directly from the lattice basis.  For
$r,s>0$,
\begin{equation}
 \operatorname{Tr}
 \bigl[P_-S^rP_-,P_-S^{-s}P_-\bigr]
 =r\,\delta_{rs}.
 \label{eq:shift-commutator}
\end{equation}
Indeed, when $r=s$ the two products differ by the projection onto the $r$
sites $-r,-r+1,\ldots,-1$ next to the boundary; when $r\ne s$ the
commutator has no diagonal part.  Consequently,
\begin{equation}
 \operatorname{Tr}[A,B]
 =E(V),
 \qquad
 E(V):=\sum_{r>0}rV_rV_{-r},
 \label{eq:edge-energy}
\end{equation}
and integration of \eqref{eq:edge-derivative} yields the exact identity
\begin{equation}
{\displaystyle
 \log D_-(t)=\frac{t^2}{2}E(V).}
 \label{eq:one-edge-quadratic}
\end{equation}
Thus a single interface is exactly quadratic in $t$, and gives half of the quadratic contribution to the finite interval case \eqref{eq:Szego exp}.  

The result of this section may be summarized as an elementary
finite-range version of the strong \szego theorem.  If
\(V=\sum_{|r|\leq R}V_rS^r\) is a scalar Laurent operator of finite
range, then
\begin{equation}
\log D_N(e^{tV})
=
NtV_0
+
t^2\sum_{r>0}rV_rV_{-r}
+
o(1),
\qquad N\to\infty ,
\label{eq:finite-range-szego}
\end{equation}
uniformly for \(t \in [0,1]\).  Thus the two boundaries decouple
in the large-\(N\) limit, while all contributions beyond quadratic
order cancel at each isolated boundary.  The appendix gives a
trace-norm formulation of the decoupling argument; extending this
elementary proof from finite-range operators to the full
Szeg\H{o}-Krein class requires additional control of the long-range
tails.

\section{\szego from a Fermi sea picture}
In our second derivation, we will resort to the physics of Fermions, where determinants occur naturally in many problems. In particular, we start by representing our desired determinant as a Fermion expectation value. Let us define fermions $\psi_k$ on the infinite lattice with the usual anti-commutation relations $[\psi^\dagger_n,\psi_m]_{+}=\delta_{n,m}$, and define a Fermi sea state,
\begin{equation}
    \ket{\Omega} = \prod_{k \in \Omega} \psi^\dagger_k \ket{0}
\end{equation}
where $\ket{0}$ is the vacuum with no fermions. We define a Hamiltonian operator out of the matrix $V$ and the fermions,
\begin{equation}
    H = \sum_{m,n \in \ZZ} V_{m n} \psi^\dagger_m \psi_n
\end{equation}
We then have the standard slater determinant formula,
\begin{equation}
    D_\Omega(t)=\bra{\Omega} e^{t H} \ket{\Omega} = \det(\p e^{t V} \p)
\end{equation}
Therefore,
\begin{equation}
    \log D_\Omega(t) = \log \bra{\Omega} e^{t H} \ket{\Omega} = \sum_{n \geq 1} \frac{t^n}{n!} \expval{H^n}_c
\end{equation}
where $\expval{H^n}_c$ is called the $n^{\text{th}}$ cumulant of the  operator $H$. Using the fact that,
\begin{equation}
    \bra{\Omega} \psi^\dagger_m \psi_n \ket{\Omega} = (\p)_{m n} \quad \text{and} \quad \bra{\Omega} \psi_m \psi^\dagger_n \ket{\Omega} = (\q)_{m n},
\end{equation}
We can verify that the cumulants of the Hamiltonian match with the Taylor expansion of $\log D_\Omega(t)$ term by term. 

When the size of the Fermi interval is large, the modes near the Fermi points at ends of the interval $\Omega$ decouple. The contribution to the expectation value due to the interaction between modes near different Fermi points is suppressed by the decay of $V$. As stressed before, we do not need the translational invariance of $V$ here.

When the Fermi points decouple, the expectation value
of $e^{t H}$ in the state $\ket{\Omega}$ can be approximated to be the product of two expectation values,
\begin{equation}
    \bra{\Omega} e^{t H} \ket{\Omega} \approx \expval{e^{t H_L}}_L \expval{e^{t H_R}}_R
\end{equation}
where $H_{R,L}$ are the decoupled Hamiltonians for the right/left modes, and the expectation values are evaluated on the Dirac sea state, where one side of the lattice is fully filled. 
We approximately evaluate each of the expectation values above. We will now assume that $V$ is translationally invariant. We have,
\begin{equation}
    H_R \approx \sum_{r \in \ZZ} V_r \sum_{n \in \ZZ} R^\dagger_{n+r} R_n 
\end{equation}
Where $R_n$ are fermions near the right edge. Defining normal ordering with respect to the Dirac sea, and normal ordered current operator,
\begin{equation}
    J_r = \sum_{n \in \ZZ} :R^\dagger_{n+r} R_n :
\end{equation}
we can write,
\begin{equation}
    H_R = \; :H_R: + \expval{H_R}_R = \sum_{r \neq 0} V_r J_r + \frac{1}{2}\abs{\Omega}V_0
\end{equation}
$\frac{1}{2}\abs{\Omega}$ is the natural regulator for the expectation value of $H_R$. It is easy to verify from Wick's theorem the algebra of the normal ordered currents to be,
\begin{equation}
    [J_r, J_s] = s \delta_{r+s,0}
\end{equation}
Using the fact that $J_{r<0}$ acting on the right Dirac vacuum is zero, and the standard Baker-Campbell-Hausdorff lemma, we get,
\begin{equation}
    \expval{e^{t \; :H_R:}}_R = \expval{\exp(t\sum_{r>0} V_r J_r + t\sum_{s<0} V_s J_s
    )}_R = \exp(\frac{t^2}{2} \sum_{r>0} r V_r V_{-r})
\end{equation}
This result is true for all $t$. 
Putting together the left and right mode contributions, we have,
\begin{equation}
    \log D_\Omega(t) \approx t \abs{\Omega} V_0 + t^2 \sum_{r>0} r V_r V_{-r} + o(1)
\end{equation}
All higher order cumulants are vanishing in this approximation. This derivation is essentially a re-writing of the matrix argument presented in the last section to the language of fermions.

\section{Numerical illustrations}
\label{sec:numerics}

We now illustrate several features of the boundary picture: convergence
inside the strong Szeg\H{o} regime, logarithmic growth Fisher-Hartwig regime,
power-law growth for still slower decay, and the  behavior
with real-space disorder.  Throughout, we subtract the extensive term and write
\begin{equation}
C_N=\log\det T_N(e^V)-NV_0 .
\label{eq:numerical-CN}
\end{equation}
For a translation-invariant matrix, its finite-$N$ quadratic boundary
contribution is obtained from \eqref{eq:FiniteNCorrection}:
\begin{equation}
E_N(V)=\sum_{r\geq1}\min(r,N)V_rV_{-r}.
\label{eq:numerical-EN}
\end{equation}

\subsection{Slow but sufficient algebraic decay}
\label{subsec:numerics-sufficient-decay}

To make convergence deliberately slow while remaining inside the strong
Szeg\H{o} class, we take
\begin{equation}
V_0=0,
\qquad
V_r=V_{-r}=\frac{\lambda}{r^p},
\qquad
\lambda=0.35,
\quad p=1.2.
\label{eq:numerical-power-law}
\end{equation}
The strong Szeg\H{o} constant is finite,
\begin{equation}
E(V)=\lambda^2\sum_{r\geq1}r^{1-2p}
=\lambda^2\zeta(2p-1)
=0.38042954\ldots .
\label{eq:numerical-power-law-E}
\end{equation}
On the other hand,
\begin{equation}
E(V)-E_N(V)\sim
\frac{\lambda^2}{(2p-2)(2p-1)}N^{2-2p},
\label{eq:numerical-power-law-rate}
\end{equation}
so for $p=1.2$ the convergence is only of order $N^{-0.4}$.

For the numerical calculation, we numerically Fourier transform sample the symbol
$\exp[V(\theta)]$,
and form the corresponding $N\times N$ Toeplitz matrix.  Fig.~\ref{fig:sufficient-decay} shows the exact correction $C_N$, the
finite-$N$ expression $E_N(V)$, and their common limiting value. 

\begin{figure}[t]
    \centering
    \includegraphics[width=0.7\linewidth]{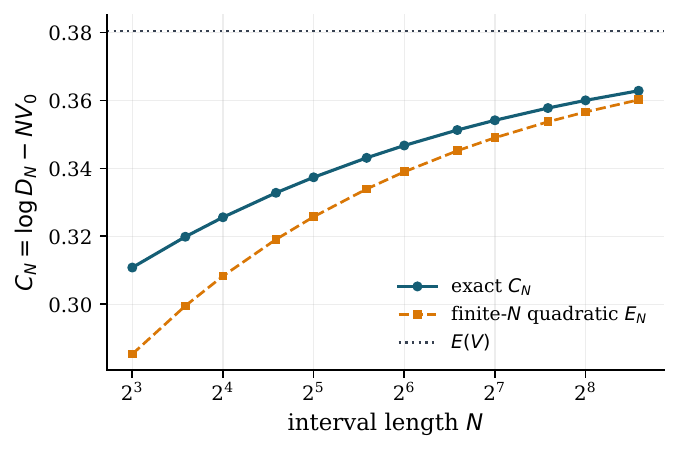}
    \caption{Slow but sufficient algebraic decay.  Both the exact
    boundary correction $C_N$ and the finite-$N$ quadratic expression
    $E_N(V)$ approach the strong Szeg\H{o} constant.  Their convergence
    is slow because $p=1.2$ lies close to the borderline value $p=1$.}
    \label{fig:sufficient-decay}
\end{figure}

\subsection{Borderline decay and logarithmic growth}
\label{subsec:numerics-borderline}

To study Fisher-Hartwig like logarithmic behavior, we take
\begin{equation}
V_0=0,
\qquad
V_r=V_{-r}=\frac{\alpha}{r},
\qquad
\alpha=0.30.
\label{eq:numerical-borderline}
\end{equation}
Then
\begin{equation}
V(\theta)
=-2\alpha\log\left|2\sin\frac{\theta}{2}\right|,
\qquad
e^{V(\theta)}
=\left|2\sin\frac{\theta}{2}\right|^{-2\alpha}.
\label{eq:numerical-borderline-symbol}
\end{equation}
The Szeg\H{o} sum no longer converges.  Its finite-$N$ form already
predicts
\begin{equation}
E_N(V)=\alpha^2\log N+O(1).
\label{eq:numerical-borderline-prediction}
\end{equation}
This example is especially convenient because it is a pure Fisher-Hartwig singularity of the form $f(e^{i\theta}) = \abs{1 - e^{i\theta}}^{-\alpha}$.
Its finite-$N$ Toeplitz determinant is known exactly
\cite{boettcherSilbermann1985,deiftItsKrasovsky2014}:
\begin{equation}
D_N
=
\frac{G(1-\alpha)^2}{G(1-2\alpha)}
\frac{G(N+1)G(N+1-2\alpha)}
     {G(N+1-\alpha)^2}
=
\prod_{j=0}^{N-1}
\frac{\Gamma(j+1)\Gamma(j+1-2\alpha)}
     {\Gamma(j+1-\alpha)^2},
\label{eq:numerical-borderline-product}
\end{equation}
where $G$ denotes the Barnes $G$-function.  The second equality
follows from $G(z+1)=\Gamma(z)G(z)$.
Numerical evaluation of this product confirms
\begin{equation}
\log D_N=\alpha^2\log N+O(1).
\label{eq:numerical-borderline-result}
\end{equation}
As shown in Fig.~\ref{fig:borderline-growth}, the effective slope tends
to $\alpha^2=0.09$; at $N=4096$ it is $0.090007$.

\begin{figure}[t]
    \centering
    \includegraphics[width=\linewidth]{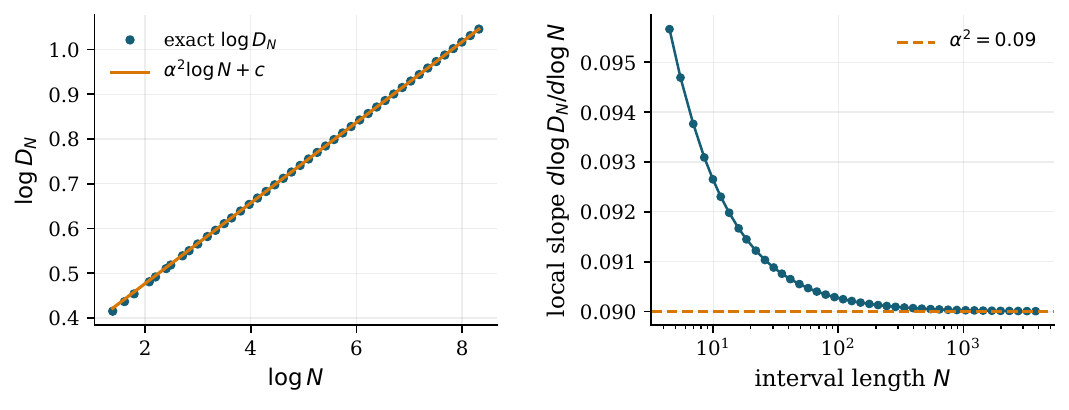}
    \caption{Logarithmic growth for the borderline coefficients
    $V_r=\alpha/r$.  The left panel plots the exact determinant against
    $\log N$.  The right panel shows the local slope approaching the
    predicted coefficient $\alpha^2=0.09$.}
    \label{fig:borderline-growth}
\end{figure}

\subsection{Slower than 1/r}
\label{subsec:numerics-too-slow}

For $1/2<p<1$, the quadratic boundary sum grows as a power rather
than logarithmically.  For $V_r=V_{-r}=\lambda/r^p$,
\begin{equation}
E_N(V)\sim
\frac{\lambda^2}{(2-2p)(2p-1)}N^{2-2p}.
\label{eq:numerical-too-slow-prediction}
\end{equation} 
we choose $p=0.75$ and $\lambda = -0.05$ for the numerical computation.
The negative sign is important: for the opposite sign,
$V(\theta)$ diverges positively as $\theta\to0$ and $e^{V(\theta)}$
is not integrable.  Figure~\ref{fig:too-slow-clean} plots the exact
correction and its quadratic approximation against
$N^{2-2p}=N^{1/2}$.  Both grow with $N$ with the same power, but they no longer coincide.
While the quadratic boundary contribution correctly diagnoses the breakdown of a finite strong \szego constant, the coefficient of the divergence is no longer fixed by the quadratic term alone, unlike in the borderline divergence case.

\begin{figure}[t]
    \centering
    \includegraphics[width=.68\linewidth]{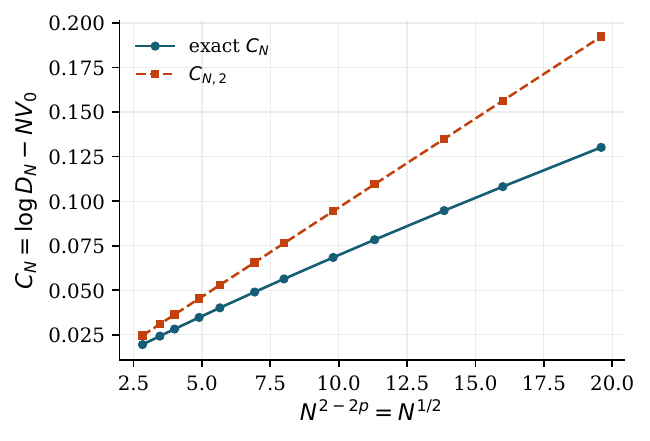}
    \caption{Example with decay slower than the borderline.
    For $p=0.75$, the quadratic boundary sum grows as $N^{1/2}$.
    The exact correction displays the same power-law growth over the
    accessible range but differs visibly from the quadratic expression,
    showing that higher-order boundary processes are no longer negligible.}
    \label{fig:too-slow-clean}
\end{figure}

\subsection{Adding Disorder}
\label{subsec:numerics-disorder}

To separate boundary decoupling from the special Toeplitz cancellation,
consider the real symmetric random tridiagonal operator
\begin{equation}
V^{(\omega)}_{n,n}=a_n,
\qquad
V^{(\omega)}_{n,n+1}=V^{(\omega)}_{n+1,n}=J_n,
\label{eq:numerical-disorder-model}
\end{equation}
where
\begin{equation}
a_n=\sigma_a\xi_n,
\qquad
J_n=J(1+\sigma_J\eta_n),
\qquad
J=0.55,
\quad \sigma_a=0.50,
\quad \sigma_J=0.40,
\label{eq:numerical-disorder-parameters}
\end{equation}
and $\xi_n,\eta_n$ are independent and uniform on $[-1,1]$.  We use the
relative boundary correction
\begin{equation}
C_N^{(\omega)}
=\log\det\bigl(P_Ne^{V^{(\omega)}}P_N\bigr)
-\Tr(P_NV^{(\omega)}P_N).
\label{eq:numerical-disorder-CN}
\end{equation}
For a fixed realization, we also set the bond at the midpoint of the
interval to zero.  The resulting operator $e^{V_{\mathrm{cut}}^{(\omega)}}$ is a direct sum of a left-edge
problem and a right-edge problem; denote its correction by
\begin{equation}
C_{N,\mathrm{cut}}^{(\omega)}
=\log\det\bigl(P_Ne^{V_{\mathrm{cut}}^{(\omega)}}P_N\bigr)
-\Tr(P_N  V_{\mathrm{cut}}^{(\omega)} P_N).
\label{eq:numerical-disorder-CNcut}
\end{equation}
Thus
\begin{equation}
C_N^{(\omega)}-C_{N,\mathrm{cut}}^{(\omega)}
\label{eq:numerical-disorder-decoupling-error}
\end{equation}
directly measures the residual communication between the two edges.
Figure~\ref{fig:disorder-fixed} shows that this difference rapidly falls
to numerical roundoff.  At $t=1$, the second-order contribution is simply
\begin{equation}
C_{N,2}^{(\omega)}
=\frac12\Tr\left[
P_N\bigl(V^{(\omega)}\bigr)^2P_N
-\bigl(P_NV^{(\omega)}P_N\bigr)^2
\right]
=\frac12\left(J_{-1}^2+J_{N-1}^2\right).
\label{eq:numerical-disorder-quadratic}
\end{equation}
This quadratic contribution does not coincide with the exact edge
contribution.  The survival of
higher edge terms is consistent with the fact that the exact
single-interface cancellation used in the Toeplitz problem is no longer
available.
\begin{figure}[H]
    \centering
    \includegraphics[width=.78\linewidth]{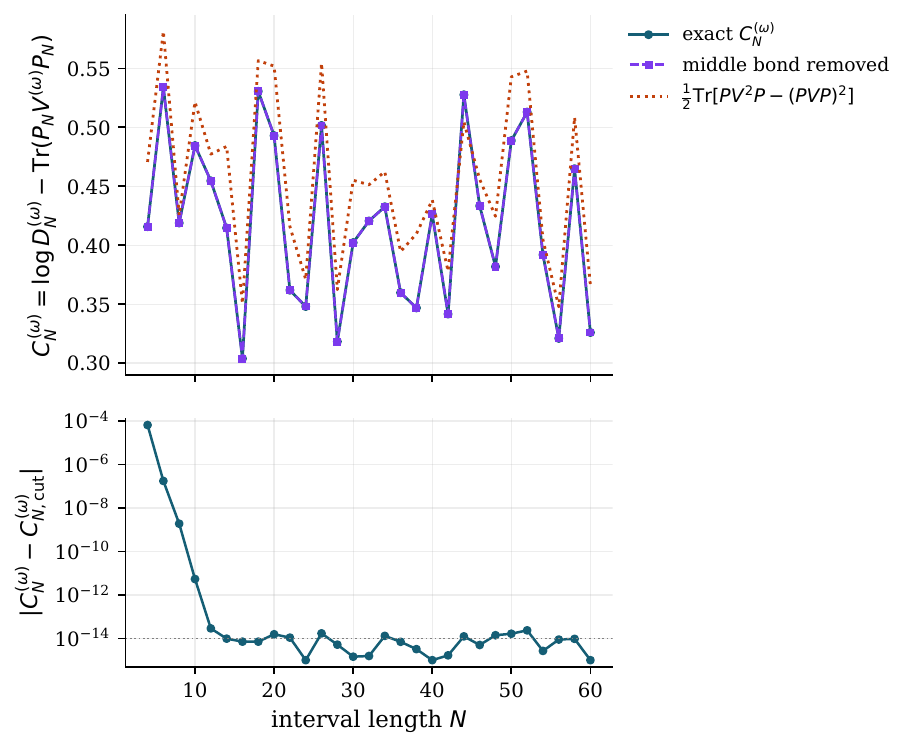}
    \caption{One fixed disorder realization.  The correction for the
    full interval is indistinguishable from the direct sum obtained by
    removing the middle bond once the two edges are sufficiently far
    apart.  The lower panel displays the rapidly decreasing decoupling
    error.  The quadratic contribution
    $\frac12\Tr[P_NV^2P_N-(P_NVP_N)^2]$ remains only an approximation.}
    \label{fig:disorder-fixed}
\end{figure}
Figure~\ref{fig:disorder-average} shows an average over $240$
independent realizations.  The mean correction becomes independent of
$N$ once the two boundary regions are separated.  The sample variance
also approaches an $N$-independent value, reflecting the local disorder
at the two edges.
For the parameters above, the large-$N$ sample means are
\begin{equation}
\overline{C_N^{(\omega)}}\simeq0.3297,
\qquad
\overline{C_{N,2}^{(\omega)}}\simeq0.3294.
\label{eq:numerical-disorder-means}
\end{equation}
{\it Numerical procedure.}
We evaluated these finite-range disordered examples in real space.  For each realization,
$V^{(\omega)}$ was represented by a finite tridiagonal matrix on an interval
containing the $N$-site region of interest and buffers on both sides.  We
diagonalized this larger matrix, constructed its exponential, restricted it
to the central region, and evaluated the relative boundary correction.  For
the fixed realization we used $28$ buffer sites on each side, giving
$L=N+56$.  For the disorder average we used $L=140$ and extracted intervals
with $N\leq80$, leaving at least $30$ exterior sites on either side.  Since
$V^{(\omega)}$ is nearest-neighbor, the influence of the artificial outer
boundaries is negligible for these buffer sizes.

\begin{figure}[H]
    \centering
    \includegraphics[width=.72\linewidth]{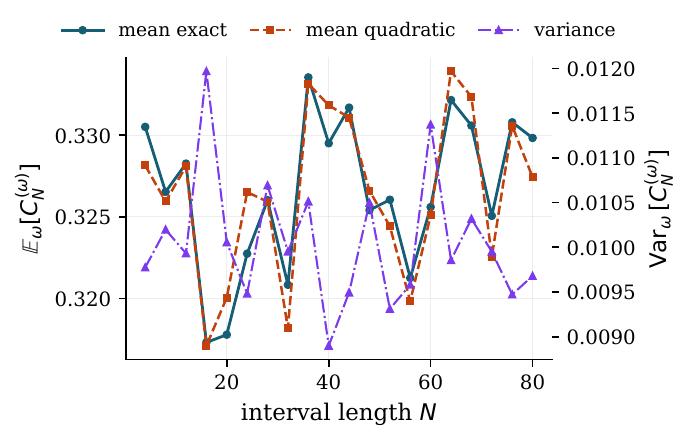}
    \caption{Disorder average over $240$ realizations.  The exact mean,
    quadratic mean, and sample variance all become independent of the
    interval length once the two boundaries are separated.}
    \label{fig:disorder-average}
\end{figure}

\subsection{Long-range disorder at and beyond the borderline}
\label{subsec:numerics-long-range-disorder}

We finally introduce real-space disorder with algebraically decaying
matrix elements.  For $i<j$, let
\begin{equation}
V_{ij}^{(\omega)}=V_{ji}^{(\omega)}
=\frac{g\,\xi_{ij}}{|i-j|^p},
\qquad
V_{ii}^{(\omega)}=0,
\qquad g=0.10,
\label{eq:numerical-long-range-disorder}
\end{equation}
where the $\xi_{ij}$ are independent standard normal variables.  The
ensemble-averaged quadratic contribution in the infinite system is
\begin{equation}
\overline{C_{N,2}^{(\omega)}}
=g^2\sum_{r\geq1}\frac{\min(r,N)}{r^{2p}}.
\label{eq:numerical-long-range-disorder-quadratic}
\end{equation}
Consequently,
\begin{equation}
\overline{C_{N,2}^{(\omega)}}
\sim
\begin{cases}
g^2\log N, & p=1,\\[2mm]
\displaystyle
\frac{g^2}{(2-2p)(2p-1)}N^{2-2p},&1/2<p<1.
\end{cases}
\label{eq:numerical-long-range-disorder-scaling}
\end{equation}

{\it Numerical details.}
For each realization we generated the full real symmetric matrix on
$L=1024$ sites and diagonalized it as
$V^{(\omega)}=U\Lambda U^{\mathsf T}$.  We then formed
$e^{V^{(\omega)}}=Ue^{\Lambda}U^{\mathsf T}$ and extracted its central
$N\times N$ principal block for $N\leq64$.  Its logarithmic determinant
gives $C_N^{(\omega)}$ because the diagonal of $V^{(\omega)}$ vanishes.
The quadratic comparison was evaluated independently from
$\frac12\Tr[P_NV^2P_N-(P_NVP_N)^2]$.  Centering the interval leaves a
buffer of at least $480$ sites on either side at the largest $N$ and
therefore suppresses effects of the artificial outer boundaries.
Figure~\ref{fig:long-range-disorder}
shows averages over $80$ realizations.  At $p=1$ the mean correction is
approximately linear in $\log N$, while at $p=0.75$ it is approximately
linear in $N^{1/2}$.  The quadratic boundary expression works remarkably
well even though the matrix is both disordered and long-ranged: at
$N=64$ its relative difference from the exact ensemble mean is about
$0.18\%$ for $p=1$ and $0.46\%$ for $p=0.75$.  The agreement is expected
to be best at the weak disorder strength $g=0.10$ used here, because the
first omitted terms are of higher order in $g$.  More specifically, the
Gaussian ensemble is invariant under $V^{(\omega)}\mapsto-V^{(\omega)}$,
so all odd powers vanish after disorder averaging and the first correction
to the quadratic mean is of order $g^4$.  Thus disorder does not
remove the logarithmic or power-law boundary growth; it changes the
sample-to-sample fluctuations while the averaged quadratic boundary
count continues to capture the leading behavior.

\begin{figure}[H]
    \centering
    \includegraphics[width=\linewidth]{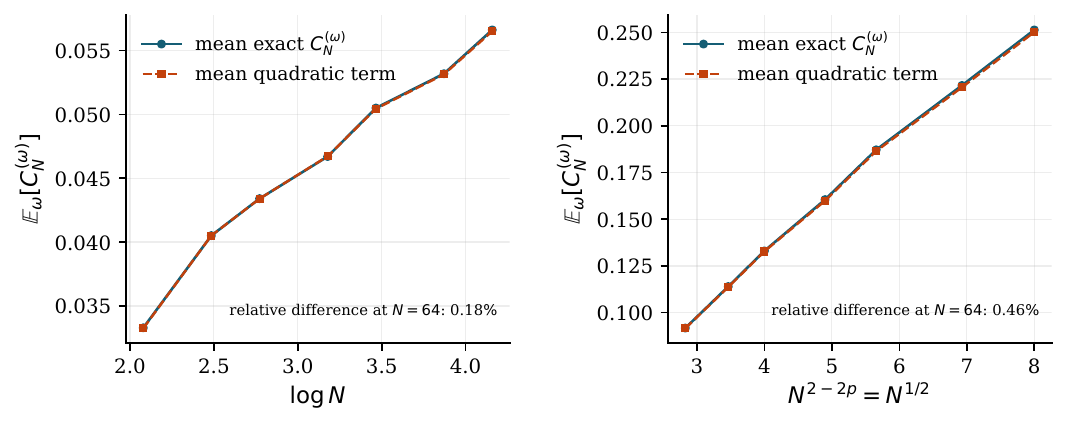}
    \caption{Long-range disorder at and beyond the borderline.  The left
    panel shows logarithmic growth for rms matrix elements proportional
    to $1/|i-j|$.  The right panel shows $N^{1/2}$ growth when the decay
    is slowed to $1/|i-j|^{0.75}$.  The exact and quadratic ensemble means
    differ by less than $0.5\%$ at $N=64$ for $g=0.10$.}
    \label{fig:long-range-disorder}
\end{figure}

\section{Summary and outlook}
\label{sec:summary}

We have presented an elementary boundary interpretation of the strong Szeg\H{o} formula. Restricting an infinite operator to a finite interval creates two interfaces, and the subleading part of the logarithmic determinant measures the leakage through these interfaces. At second order, the factor $r$ in
\[
\sum_{r>0}rV_rV_{-r}
\]
simply counts the sites from which a jump of length $r$ can cross a boundary. The full result follows by combining two facts: the two boundaries decouple when their separation becomes large, and the determinant associated with a single interface of a scalar Laurent operator is exactly quadratic. The first fact is a general locality statement, whereas the second uses translation invariance and the commutativity of the scalar Laurent coefficients.

The same picture was obtained from free fermions. In that formulation, the finite interval is a filled Fermi sea with two Fermi edges. Their decoupling converts the determinant into a product of two single-edge expectation values, and the current algebra makes each of these expectation values Gaussian. The matrix and fermionic derivations therefore describe the same mechanism in two languages: the strong Szeg\H{o} term is the finite contribution produced by two independent boundaries.

The numerical examples illustrate both the reach and the limitations of the boundary picture.  When the Fourier coefficients decay sufficiently rapidly, the correction
\[
C_N=\log D_N-NV_0
\]
approaches the strong Szeg\H{o} constant, although this convergence can be extremely slow.  At the borderline $V_r\sim 1/r$, the quadratic boundary sum grows as $\log N$ and correctly reproduces the leading Fisher-Hartwig coefficient.  When the decay is still slower, the correction grows as a power of $N$.  In the clean example studied here, the quadratic term captures the observed power law but not its coefficient, as higher order boundary contributions are no longer negligible.  Thus the boundary expansion remains useful beyond the Szeg\H{o}-Krein regime, but it no longer reduces to the universal quadratic constant of the strong Szeg\H{o} theorem.

The disordered examples separate two features that coincide in the Toeplitz problem.  For finite-range disorder, the two boundaries still decouple once they are sufficiently far apart, confirming that boundary decoupling is a geometric and quasi-local phenomenon that does not require translation invariance.  The contribution of an individual boundary, however, is generally not exactly quadratic and depends on the local disorder realization.  Remarkably, for the weak centered Gaussian ensembles considered here, the disorder-averaged quadratic term remains extremely accurate over the accessible sizes even at and beyond the Szeg\H{o}-Krein threshold.  This agreement is enhanced because the disorder distribution is invariant under \(V\mapsto -V\), so all odd orders vanish after disorder averaging and that the first correction to the quadratic mean is of fourth order in the disorder strength.  It should not be interpreted as a generic property of clean slowly decaying symbols, nor as evidence of uniform accuracy as $N\to\infty$.

Several extensions appear natural. For disordered operators, one may study the complete probability distribution of the boundary correction, its correlations between separated boundaries, and the conditions under which it becomes self-averaging.  For finite-range systems these quantities are determined by disorder in a fixed boundary layer.  For long-range systems the relevant boundary region is no longer strictly local, and an important question is how its effective width grows with $N$ and with the decay exponent.  It would also be useful to determine analytically when higher-order terms merely renormalize the coefficient predicted by the quadratic boundary count and when they change the leading scaling itself.

In higher dimensions, the two endpoints of an interval are replaced by an extended boundary, suggesting surface area contributions together with additional terms associated with corners, curvature, or singular momentum space structure.  The same viewpoint may provide an elementary route to Sobolev-Widom type asymptotics and their disordered analogues.  Finally, matrix-valued symbols and internal degrees of freedom introduce noncommuting boundary processes.  Understanding which of the scalar cancellations survive in this setting may connect the present approach to multiband free-fermion systems, entanglement determinants, and topological boundary phenomena.

Another important direction is the extension to block Toeplitz
operators describing spin, orbital, sublattice, or Nambu degrees of
freedom. In this case the Fourier coefficients are matrices and need
not commute. Boundary decoupling should remain valid under suitable
locality assumptions, but the single-interface
Baker-Campbell-Hausdorff reduction is no longer generally quadratic:
nested commutators can produce genuine higher-order edge terms. A
recent two-band free-fermion example was studied by Jordan, Gangardt,
and Abanov \cite{JordanGangardtAbanov2026}. In that problem the
thermodynamic return amplitude is a block Szeg\H{o}-Widom constant,
and its logarithm contains both a quadratic contribution and a bounded
theta-function correction. It would be useful to understand more
generally which algebraic properties of a matrix-valued symbol
determine the resulting boundary functional and under what conditions
the scalar quadratic formula survives.

\section*{Acknowledgements}
IK would like to thank Alexander Abanov for discussions. RR thanks Diana Vaman and Joshua Bautista for useful discussions. 

{\it Use of AI-Assisted Tools}.
Large language models (ChatGPT) assisted with text editing and code development. All mathematical derivations, scientific claims, and code were independently verified by the authors.

\appendix

\section{Finite-range boundary decoupling}
\label{app:boundary-decoupling}

Let $V$ be a bounded operator on $\ell^2(\mathbb Z)$ of finite
range $R$:
\begin{equation}
V_{ij}=0
\qquad\text{whenever}\qquad
|i-j|>R.
\label{eq:decoupling-finite-range}
\end{equation}
No translation invariance is assumed. Let $P_N$ denote the
projection onto $\{0,\ldots,N-1\}$, and define
\begin{equation}
W_N=P_NVP_N.
\label{eq:decoupling-WN}
\end{equation}
The relative boundary operator is
\begin{equation}
\mathcal R_N(t)
=
\bigl(P_Ne^{tV}P_N\bigr)e^{-tW_N},
\label{eq:decoupling-relative-matrix}
\end{equation}
regarded as an operator on $P_N\ell^2(\mathbb Z)$.

We compare $\mathcal R_N(t)$ with the corresponding single-boundary
problems. Introduce the half-line projections
\begin{equation}
P=\mathbf 1_{\{j\geq0\}},
\qquad
P^{(N)}=\mathbf 1_{\{j<N\}}.
\label{eq:decoupling-halfline-projections}
\end{equation}
They commute and satisfy
\begin{equation}
P_N=PP^{(N)},
\qquad
I-P_N=P^\perp+\bigl(P^{(N)}\bigr)^\perp.
\label{eq:decoupling-complement-split}
\end{equation}
The two terms in the second identity project onto the disjoint exterior
half-lines $\{j<0\}$ and $\{j\geq N\}$, respectively.
Set
\begin{equation}
W_L=PVP,
\qquad
W_R^{(N)}=P^{(N)}VP^{(N)}.
\label{eq:decoupling-halfline-compressions}
\end{equation}
The isolated left and right-boundary corrections are the trace-class
operators
\begin{equation}
\mathcal E_L(t)
=
\int_0^t
Pe^{sV}P^\perp VPe^{-sW_L}\,ds
\label{eq:decoupling-left-edge-operator}
\end{equation}
and
\begin{equation}
\mathcal E_R^{(N)}(t)
=
\int_0^t
P^{(N)}e^{sV}\bigl(P^{(N)}\bigr)^\perp
VP^{(N)}e^{-sW_R^{(N)}}\,ds.
\label{eq:decoupling-right-edge-operator}
\end{equation}
They are trace class because $P^\perp VP$ and
$\bigl(P^{(N)}\bigr)^\perp VP^{(N)}$ have rank at most $R$.
Next we decompose into left and right boundary contributions:
\begin{lemma}[Exact boundary decomposition]
\label{lem:exact-boundary-decomposition}
For every $N$ and $t$,
\begin{equation}
\begin{split}
\mathcal R_N(t)-I
={}&
\int_0^t
P_Ne^{sV}P^\perp VP_Ne^{-sW_N}\,ds
+
\int_0^t
P_Ne^{sV}\bigl(P^{(N)}\bigr)^\perp
VP_Ne^{-sW_N}\,ds.
\end{split}
\label{eq:decoupling-boundary-split}
\end{equation}
\end{lemma}

\begin{proof}
Define
\begin{equation}
F_N(s)
=
\bigl(P_Ne^{sV}P_N\bigr)e^{-sW_N}.
\label{eq:decoupling-FN}
\end{equation}
Since $W_N=P_NVP_N$, differentiation gives
\begin{equation}
\begin{split}
\frac{d}{ds}F_N(s)
&=
P_Ne^{sV}VP_Ne^{-sW_N}
-
P_Ne^{sV}P_NVP_Ne^{-sW_N}
\\
&=
P_Ne^{sV}(I-P_N)VP_Ne^{-sW_N}.
\end{split}
\label{eq:decoupling-duhamel-derivative}
\end{equation}
Using \eqref{eq:decoupling-complement-split} and integrating from
$0$ to $t$, with $F_N(0)=I$, proves
\eqref{eq:decoupling-boundary-split}.
\end{proof}

We show that the left and right contributions above can be approximated by the isolated edges ``semi infinite"  contributions $\mathcal E_L(t),\mathcal E_R^{(N)}(t)$ \eqref{eq:decoupling-left-edge-operator} and \eqref{eq:decoupling-right-edge-operator}

\begin{lemma}[Finite-propagation estimate]
\label{lem:finite-propagation}
Define
\begin{equation}
q_N=
\left\lceil
\frac{N-R+1}{R}
\right\rceil.
\label{eq:decoupling-qN}
\end{equation}
For every $T>0$, there exists a constant $C_T$, independent of
$N$, such that
\begin{equation}
\begin{split}
\sup_{|t|\leq T}
\Bigg\|
&
\int_0^t
P_Ne^{sV}P^\perp VP_Ne^{-sW_N}\,ds
-
P_N\mathcal E_L(t)P_N
\Bigg\|_1
\leq
C_T
\sum_{k\geq q_N}
\frac{(T\|V\|)^k}{k!}.
\end{split}
\label{eq:decoupling-left-trace-bound}
\end{equation}
An analogous estimate holds at the right boundary:
\begin{equation}
\begin{split}
\sup_{|t|\leq T}
\Bigg\|
&
\int_0^t
P_Ne^{sV}\bigl(P^{(N)}\bigr)^\perp
VP_Ne^{-sW_N}\,ds
-
P_N\mathcal E_R^{(N)}(t)P_N
\Bigg\|_1
\leq
C_T
\sum_{k\geq q_N}
\frac{(T\|V\|)^k}{k!}.
\end{split}
\label{eq:decoupling-right-trace-bound}
\end{equation}
In particular, both differences converge to zero faster than
exponentially in $N$, uniformly for $|t|\leq T$.
\end{lemma}

\begin{proof}
We prove the left-boundary statement. The right-boundary statement
follows from the same argument with the orientation reversed.

Since $V$ has range $R$, the operator $P^\perp VP$ acts only on the
first $R$ sites of the positive half-line. Consequently, for
$N\geq R$,
\begin{equation}
P^\perp VP=P^\perp VP_N.
\label{eq:decoupling-left-support}
\end{equation}
It follows that
\begin{equation}
\begin{split}
&
\int_0^t
P_Ne^{sV}P^\perp VP_Ne^{-sW_N}\,ds
-
P_N\mathcal E_L(t)P_N
\\
&\qquad=
\int_0^t
P_Ne^{sV}P^\perp VP
\left(
e^{-sW_N}-P_Ne^{-sW_L}P_N
\right)\,ds.
\end{split}
\label{eq:decoupling-left-difference}
\end{equation}
An application of Duhamel's formula gives the exact identity
\begin{equation}
P_Ne^{-sW_L}P_N-e^{-sW_N}
=
-\int_0^s
P_Ne^{-uW_L}(P-P_N)W_LP_N
e^{-(s-u)W_N}\,du.
\label{eq:decoupling-second-duhamel}
\end{equation}
The operator $P^\perp VP$ is supported within $R$ sites of the left
boundary, whereas $P-P_N$ begins at the site $N$. Since $W_L$ has
range $R$,
\begin{equation}
P^\perp VP\,W_L^k(P-P_N)=0
\qquad
\text{for}\qquad
k<q_N.
\label{eq:decoupling-vanishing-powers}
\end{equation}
Expanding the exponential therefore yields, uniformly for
$|u|\leq T$,
\begin{equation}
\left\|
P^\perp VP\,e^{-uW_L}(P-P_N)
\right\|_1
\leq
\|P^\perp VP\|_1
\sum_{k\geq q_N}
\frac{(T\|V\|)^k}{k!}.
\label{eq:decoupling-propagation-bound}
\end{equation}
The remaining factors satisfy
\begin{equation}
\|e^{sV}\|\leq e^{T\|V\|},
\qquad
\|e^{-sW_N}\|\leq e^{T\|V\|},
\qquad
\|P^\perp VP\|_1\leq R\|V\|.
\label{eq:decoupling-uniform-bounds}
\end{equation}
Substitution into \eqref{eq:decoupling-left-difference} proves
\eqref{eq:decoupling-left-trace-bound}, after absorbing the remaining
bounded factors and the two integrations into $C_T$.

Finally,
\begin{equation}
\sum_{k\geq q}\frac{x^k}{k!}
\leq
e^x\frac{x^q}{q!}
\leq
e^x\left(\frac{ex}{q}\right)^q.
\label{eq:decoupling-factorial-tail}
\end{equation}
Since $q_N$ grows linearly with $N$, the right-hand side decays like
$\exp[-cN\log N+O(N)]$. This proves the stated
faster-than-exponential convergence.
\end{proof}

In the approximation above, the left and right single boundary parts $P_N\mathcal E_L(t)P_N$ and $P_N\mathcal E_R^{(N)}(t)P_N$ are still not completely decoupled since they work on the same space. In the next lemma we show that they can be now approximated as working on orthogonal spaces. 
We first divide the interval into two disjoint edge regions $P_N=P_{N,L}+P_{N,R}$ where
\begin{equation}
P_{N,L}
=
\sum_{j=0}^{\left\lfloor\frac{N}{2}\right\rfloor-1}|j\rangle\langle j|,
\qquad
P_{N,R}
=
\sum_{j=N-\left\lfloor\frac{N}{2}\right\rfloor}^{N-1}|j\rangle\langle j|.
\label{eq:decoupling-edge-projections}
\end{equation}

We aim to replace $P_N\mathcal E_L(t)P_N$ and $P_N\mathcal E_R^{(N)}(t)P_N$ with $P_{N,L}\mathcal E_L(t)P_{N,L}\text{ and }P_{N,R}\mathcal E_R^{(N)}(t)P_{N,R}$ that work on separate subspaces. Indeed, we have:
\begin{lemma}[Separation of the edge regions]
\label{lem:edge-separation}
For every $T>0$,
\begin{equation}
\begin{split}
\sup_{|t|\leq T}
\Bigg\|
&
\int_0^t
P_Ne^{sV}P^\perp VP_Ne^{-sW_N}\,ds
-
P_{N,L}\mathcal E_L(t)P_{N,L}
\Bigg\|_1
\longrightarrow0
\end{split}
\label{eq:decoupling-left-edge-limit}
\end{equation}
and
\begin{equation}
\begin{split}
\sup_{|t|\leq T}
\Bigg\|
&
\int_0^t
P_Ne^{sV}\bigl(P^{(N)}\bigr)^\perp
VP_Ne^{-sW_N}\,ds
-
P_{N,R}\mathcal E_R^{(N)}(t)P_{N,R}
\Bigg\|_1
\longrightarrow0
\end{split}
\label{eq:decoupling-right-edge-limit}
\end{equation}
as $N\to\infty$.
\end{lemma}

\begin{proof}
Lemma~\ref{lem:finite-propagation} first allows the left finite-interval
correction to be replaced by $P_N\mathcal E_L(t)P_N$.
Since $\mathcal E_L(t)$ is trace class and
$P_{N,L}\to P$ strongly,
\begin{equation}
\sup_{|t|\leq T}
\left\|
P_N\mathcal E_L(t)P_N
-
P_{N,L}\mathcal E_L(t)P_{N,L}
\right\|_1
\longrightarrow0.
\label{eq:decoupling-left-truncation}
\end{equation}
The convergence is uniform in $t$ because
$t\mapsto\mathcal E_L(t)$ is continuous in trace norm and
$[-T,T]$ is compact.

At the right boundary, the same finite-propagation estimate allows
$\mathcal E_R^{(N)}(t)$ to be truncated to the nearest
$\left\lfloor\frac{N}{2}\right\rfloor$ sites, uniformly in $N$. This proves
\eqref{eq:decoupling-right-edge-limit}.
\end{proof}

\begin{theorem}[Decoupling of distant boundaries]
\label{thm:boundary-decoupling}
Let $V$ be a bounded finite-range operator on
$\ell^2(\mathbb Z)$. Then, uniformly for $t$ in compact intervals,
\begin{equation}
\begin{split}
\Bigl\|
\mathcal R_N(t)-I
&-
P_{N,L}\mathcal E_L(t)P_{N,L}
-
P_{N,R}\mathcal E_R^{(N)}(t)P_{N,R}
\Bigr\|_1
\longrightarrow0.
\end{split}
\label{eq:decoupling-two-edge-limit}
\end{equation}
Consequently,
\begin{equation}
\det\mathcal R_N(t)
-
\det\!\left(I+\mathcal E_L(t)\right)
\det\!\left(I+\mathcal E_R^{(N)}(t)\right)
\longrightarrow0,
\label{eq:decoupling-general-factorization}
\end{equation}
again uniformly for $t$ in compact intervals.
\end{theorem}

\begin{proof}
Equation \eqref{eq:decoupling-two-edge-limit} follows immediately by
combining Lemma~\ref{lem:exact-boundary-decomposition} with
Lemma~\ref{lem:edge-separation}.
The two approximating edge operators act on the mutually orthogonal
subspaces selected by $P_{N,L}$ and $P_{N,R}$. Their sum is therefore
an exact direct sum, and hence
\begin{equation}
\begin{split}
\det\Bigl(
I
&+
P_{N,L}\mathcal E_L(t)P_{N,L}
+
P_{N,R}\mathcal E_R^{(N)}(t)P_{N,R}
\Bigr)
\\
&=
\det\!\left(
I+P_{N,L}\mathcal E_L(t)P_{N,L}
\right)
\det\!\left(
I+P_{N,R}\mathcal E_R^{(N)}(t)P_{N,R}
\right).
\end{split}
\label{eq:decoupling-finite-direct-sum}
\end{equation}
For trace-class operators, the Fredholm determinant satisfies
\begin{equation}
|\det(I+A)-\det(I+B)|
\leq
\|A-B\|_1
\exp\!\left(
1+\|A\|_1+\|B\|_1
\right).
\label{eq:decoupling-determinant-continuity}
\end{equation}
All trace norms appearing above are bounded uniformly in $N$ and
in $|t|\leq T$. Equation
\eqref{eq:decoupling-two-edge-limit} and the determinant-continuity
bound therefore show that $\det\mathcal R_N(t)$ differs by $o(1)$
from the product on the right-hand side of
\eqref{eq:decoupling-finite-direct-sum}.

Finally, the same trace-norm truncation used in
Lemma~\ref{lem:edge-separation} gives
\begin{equation}
\det\!\left(
I+P_{N,L}\mathcal E_L(t)P_{N,L}
\right)
\longrightarrow
\det\!\left(I+\mathcal E_L(t)\right)
\label{eq:decoupling-left-determinant-limit}
\end{equation}
and
\begin{equation}
\det\!\left(
I+P_{N,R}\mathcal E_R^{(N)}(t)P_{N,R}
\right)
-
\det\!\left(I+\mathcal E_R^{(N)}(t)\right)
\longrightarrow0.
\label{eq:decoupling-right-determinant-limit}
\end{equation}
Together these relations prove
\eqref{eq:decoupling-general-factorization}.
\end{proof}

Theorem~\ref{thm:boundary-decoupling} is purely geometric: it uses
only boundedness and finite range. Translation invariance enters only
when the isolated edge determinants are evaluated.

\begin{corollary}[Finite-range strong Szeg\H{o} formula]
\label{cor:finite-range-strong-szego}
Suppose, in addition, that $V$ is a scalar Laurent operator,
\begin{equation}
V=\sum_{|r|\leq R}V_rS^r.
\label{eq:decoupling-Laurent-specialization}
\end{equation}
Then, uniformly for $t$ in compact intervals,
\begin{equation}
D_N(e^{tV})
=
\exp\!\left(
NtV_0+
t^2\sum_{r>0}rV_rV_{-r}
\right)
\bigl(1+o(1)\bigr),
\qquad N\to\infty.
\label{eq:decoupling-strong-szego}
\end{equation}
\end{corollary}

\begin{proof}
Translation invariance makes the right-boundary problem a translated
copy of the left-boundary problem. The half-line calculation in
Sec.~\ref{sec:one-interface}, equivalently the
Helton-Howe-Pincus determinant identity
\cite{Pincus1968,HeltonHowe1973}, gives
\begin{equation}
\det\!\left(I+\mathcal E_L(t)\right)
=
\det\!\left(I+\mathcal E_R^{(N)}(t)\right)
=
\exp\!\left(
\frac{t^2}{2}
\sum_{r>0}rV_rV_{-r}
\right).
\label{eq:decoupling-edge-determinants}
\end{equation}
Theorem~\ref{thm:boundary-decoupling} therefore implies
\begin{equation}
\det\mathcal R_N(t)
\longrightarrow
\exp\!\left(
t^2\sum_{r>0}rV_rV_{-r}
\right).
\label{eq:decoupling-relative-limit}
\end{equation}
For a Laurent operator,
\begin{equation}
\operatorname{Tr}W_N=NV_0.
\label{eq:decoupling-bulk-trace}
\end{equation}
Moreover,
\begin{equation}
\det\mathcal R_N(t)
=
e^{-t\operatorname{Tr}W_N}
\det\!\left(P_Ne^{tV}P_N\right)
=
e^{-NtV_0}D_N(e^{tV}).
\label{eq:decoupling-relative-determinant}
\end{equation}
Combining
\eqref{eq:decoupling-relative-limit}--
\eqref{eq:decoupling-relative-determinant}
proves \eqref{eq:decoupling-strong-szego}.
\end{proof}

\begin{remark}[Long-range operators]
\label{rem:long-range-decoupling}
The finite-range assumption can be weakened if the matrix elements
decay sufficiently rapidly. The exact finite-propagation estimate is
then replaced by a quasi-locality estimate: propagation between
regions separated by a distance $L$ is no longer identically zero
below a fixed order, but its trace norm tends to zero as
$L\to\infty$.

For a scalar Laurent operator
\begin{equation}
V=\sum_{r\in\mathbb Z}V_rS^r,
\label{eq:long-range-Laurent-operator}
\end{equation}
a convenient sufficient condition is
\begin{equation}
\sum_{r\in\mathbb Z}|r|\,|V_r|<\infty.
\label{eq:weighted-l1-condition}
\end{equation}
Indeed, the part of $V$ crossing a fixed boundary is then trace
class, with
\begin{equation}
\|P^\perp VP\|_1
\leq
\sum_{r\in\mathbb Z}|r|\,|V_r|.
\label{eq:cross-boundary-trace-bound}
\end{equation}
Let $V^{(R)}=\sum_{|r|\leq R}V_rS^r$. Duhamel's formula gives,
for $|t|\leq T$, a truncation error bounded by a $T$-dependent
constant times
\begin{equation}
\sum_{|r|>R}|r|\,|V_r|,
\label{eq:long-range-tail}
\end{equation}
uniformly in the interval length $N$. One may therefore apply the
finite-range theorem to $V^{(R)}$, first let $N\to\infty$, and then
let $R\to\infty$.

Thus exponential decay, or algebraic decay
$V_r=O(|r|^{-p})$ with $p>2$, is more than sufficient for this
elementary decoupling argument. Condition
\eqref{eq:weighted-l1-condition} is convenient rather than optimal.
The usual Szeg\H{o}-Krein condition
\begin{equation}
\sum_{r\in\mathbb Z}|r|\,|V_r|^2<\infty
\label{eq:szego-krein-condition-remark}
\end{equation}
is weaker and generally requires additional Hankel-operator or
Borodin-Okounkov methods.
\end{remark}

\bibliographystyle{unsrt}
\bibliography{szegobib}

\end{document}